\documentclass[11pt,onecolumn,twoside]{IEEEtran}

\usepackage{cite}
\usepackage{setspace}
\usepackage{amsmath,amssymb}    
\usepackage[dvips]{graphicx}   
\usepackage{verbatim}   
\usepackage{color}      
\usepackage{subfigure}  
\usepackage{epsfig}
\usepackage{float}
\usepackage{mathtools}

\newcommand{\bi}{\begin{itemize}}
\newcommand{\ei}{\end{itemize}}
\newcommand{\beq}{\begin{equation}}
\newcommand{\eeq}{\end{equation}}
\newcommand{\bqn}{\begin{eqnarray*}}
\newcommand{\eqn}{\end{eqnarray*}}
\newcommand{\ba}{\begin{array}}
\newcommand{\ea}{\end{array}}
\newcommand{\bs}{\begin{small}}
\newcommand{\es}{\end{small}}
\newcommand{\nn}{\nonumber}

\newtheorem{theorem}{Theorem}[section]
\newtheorem{lemma}[theorem]{Lemma}
\newtheorem{proposition}[theorem]{Proposition}

\newenvironment{proof}[1][Proof]{\begin{trivlist}
\item[\hskip \labelsep {\bfseries #1}]}{\end{trivlist}}

\newcommand{\qed}{\nobreak \ifvmode \relax \else
      \ifdim\lastskip<1.5em \hskip-\lastskip
      \hskip1.5em plus0em minus0.5em \fi \nobreak
      \vrule height0.75em width0.5em depth0.25em\fi}

\DeclareMathOperator*{\argmin}{arg\,min}
\DeclareMathOperator*{\argmax}{arg\,max}

\begin{document}
\title{On Quantizer Design for Distributed Bayesian Estimation in Sensor Networks}
\author{Aditya~Vempaty\thanks{This work was supported in part by ARO under Award W911NF-12-1-0383, AFOSR under Award FA9550-10-1-0458, and NSF under Award 1218289. Part of this work was presented at the 38th IEEE International Conference on Acoustics, Speech, and Signal Processing (ICASSP-2013).

The authors are with the Department of Electrical Engineering and Computer Science, Syracuse University, Syracuse, NY USA. 

(email: avempaty@syr.edu; hhe02@syr.edu; bichen@syr.edu; varshney@syr.edu)}, Hao~He, Biao~Chen, and Pramod~K.~Varshney}
\date{}
\maketitle

\begin{abstract}
We consider the problem of distributed estimation under the Bayesian criterion and explore the design of optimal quantizers in such a system. We show that, for a conditionally unbiased and efficient estimator at the fusion center and when local observations have identical distributions, it is optimal to partition the local sensors into groups, with all sensors within a group using the same quantization rule. When all the sensors use identical number of decision regions, use of identical quantizers at the sensors is optimal. When the network is constrained by the capacity of the wireless multiple access channel over which the sensors transmit their quantized observations, we show that binary quantizers at the local sensors are optimal under certain conditions. Based on these observations, we address the location parameter estimation problem and present our optimal quantizer design approach. We also derive the performance limit for distributed location parameter estimation under the Bayesian criterion and find the conditions when the widely used threshold quantizer achieves this limit. We corroborate this result using simulations. We then relax the assumption of conditionally independent observations and derive the optimality conditions of quantizers for conditionally dependent observations. Using counter-examples, we also show that the previous results do not hold in this setting of dependent observations and, therefore, identical quantizers are not optimal.
\end{abstract}


\section{Introduction}
\label{intro}
Distributed parameter estimation from quantized data has been an active area of research \cite{reibman_tcomm93,Marano&etal:07sp,chen_tsp10,chen_tsp10_2,wu:one-bit}. In a typical distributed estimation framework\footnote{In the literature, the terms `distributed' and `decentralized' have often been used interchangeably. In this paper, we use the term `distributed' and it refers to the case when the local sensors perform local processing before sending the data to a central unit.}, local sensors send their data to a fusion center. At the fusion center, an estimation algorithm is applied to estimate the unknown parameter based on the data received from different local sensors. However, due to bandwidth/energy constraints, local observations are often quantized before they are transmitted to the fusion center. Identical quantizers at the sensors have traditionally been used by researchers as it simplifies the design problem \cite{chen_tsp10} \cite{kar_tsp12}. However, relatively little is known about the optimality of these identical quantizers. For decentralized detection, Tsitsiklis \cite{Tsitsiklis88} showed the asymptotic optimality of identical quantizers with conditionally independent and identically distributed sensor observations. In \cite{reibman_tcomm93}, the authors considered the design of optimal quantizers for distributed estimation under different distortion criteria. Using the minimax criterion, optimal quantizers have been found in \cite{wu:one-bit} and \cite{chen_tsp10_2}. The maximum likelihood estimator has been used at the Fusion Center (FC) in \cite{swami_icisip05} for which the optimal quantizers have been shown to be the score functions which depend on the true value of the parameter. A discussion on the design of quantizers with design goals of bandwidth efficiency, scalability, and robustness to network changes can be found in \cite{xiao&etal:06spm}. In \cite{gubner_tit93}, an algorithm was developed for the design of a non-linear multiple-sensor distributed estimation system by partitioning the real line for quantization. 

When considering the problem of distributed parameter estimation in sensor networks, besides the energy constraints, we need to be aware of the communication limitations of the network. The amount of information from each sensor is limited by the number of bits it transmits to the fusion center. However, finite channel throughput restricts the number of bits which the sensors can transmit to the fusion center. Chamberland and Veeravalli \cite{Chamberland&Veeravalli:SP03} have addressed the problem of decentralized detection in sensor networks under such a rate-constraint. Ribeiro and Giannakis \cite{gian_tsp06_1,gian_tsp06_2} have also addressed the problem of bandwidth-constrained distributed estimation in wireless sensor networks. However, most of the above works either deal with the case of identical quantizers, or consider the case of estimating a deterministic unknown parameter where the optimal quantizer depends on the unknown itself. In our work, we find the optimality conditions under the widely used assumption of identical quantizers. We also address quantizer design for the
Bayesian setup where average distortion is considered as the cost function, and the optimal quantizers are not dependent on the unknown.

Building on our preliminary work \cite{Vempaty_icassp13_opt}, in this paper, we study the problem of quantizer design in a distributed Bayesian estimation system. The major contributions of this work can be summarized as follows:

\bi
\item We derive the optimality conditions for an arbitrary cost function when the observations are conditionally independent. For an efficient and conditionally unbiased estimator at the fusion center, we show that it is optimal to partition the set of sensors into groups, with each group using an identical quantizer. 
\item We study quantizer design for distributed estimation under a bit rate constraint in the sensor network and determine the conditions under which it is optimal for the sensors to use binary quantizers. For the case of Gaussian observations, we show that the conditions are satisfied in the low signal-to-noise ratio (SNR) regime. 
\item We consider the location parameter estimation problem and design the optimal binary quantizer using calculus of variations. We evaluate the performance limit of such a system and derive the conditions under which the threshold quantizer attains this performance limit. 
\item We also consider the dependent observation model and derive the optimality conditions by using a hierarchical dependence framework.
\ei

The remainder of the paper is organized as follows: In Sec.~\ref{sec:prob}, we describe the distributed estimation model used in the paper and formulate the optimization problem mathematically. We derive the optimality conditions on the quantizers for an arbitrary cost function under the assumption of conditionally independent observations in Sec.~\ref{sec:opt_cond_iid}. In Sec.~\ref{sec:unbiased}, we prove that it is optimal to partition the set of sensors into groups using identical quantizers for conditionally unbiased and efficient estimators. We shift our attention to a capacity constrained wireless sensor network in Sec.~\ref{sec:rate} and determine the conditions under which binary quantizers are optimal.  In Sec.~\ref{sec:location}, we consider the location parameter estimation problem and design the optimal binary quantizer. We relax the assumption of conditionally independent observations in Sec.~\ref{sec:dep} and derive the optimality conditions. Concluding remarks are provided in Sec.~\ref{sec:conc}.

\section{Problem formulation}
\label{sec:prob}
Consider a distributed estimation problem where the goal is to estimate a random scalar parameter $\theta$ at the fusion center (FC). The parameter $\theta$ has a prior probability density function
(pdf) $p(\theta)$ where $\theta \in \Theta$. As shown in Fig. \ref{fig:model}, there are a total of $N+1$ sensors $S_0, S_1,\cdots, S_N$ in the network and sensor $S_0$ plays the role of FC whereas the other $N$ sensors are peripheral sensors. Each sensor $S_i$, for $i = 0, 1,\cdots,N$ receives a local observation $Y_i$ which is a noisy realization of the parameter $\theta$ and takes values in a set $\mathcal{Y}_i$. We assume that the joint distribution of $\textbf{Y}=[Y_0,Y_1,\cdots,Y_N]$ conditioned on $\theta$ is known to the FC for all $\theta$. In this paper, until Sec. \ref{sec:dep}, we assume that $Y_i$'s are conditionally independent and identically distributed, hence the overall likelihood function is $p(\textbf{y}|\theta)=\prod_{i=0}^N p(y_i|\theta)$. 

\begin{figure}[htb]
\centering
\includegraphics[width=2.5in,height=!]{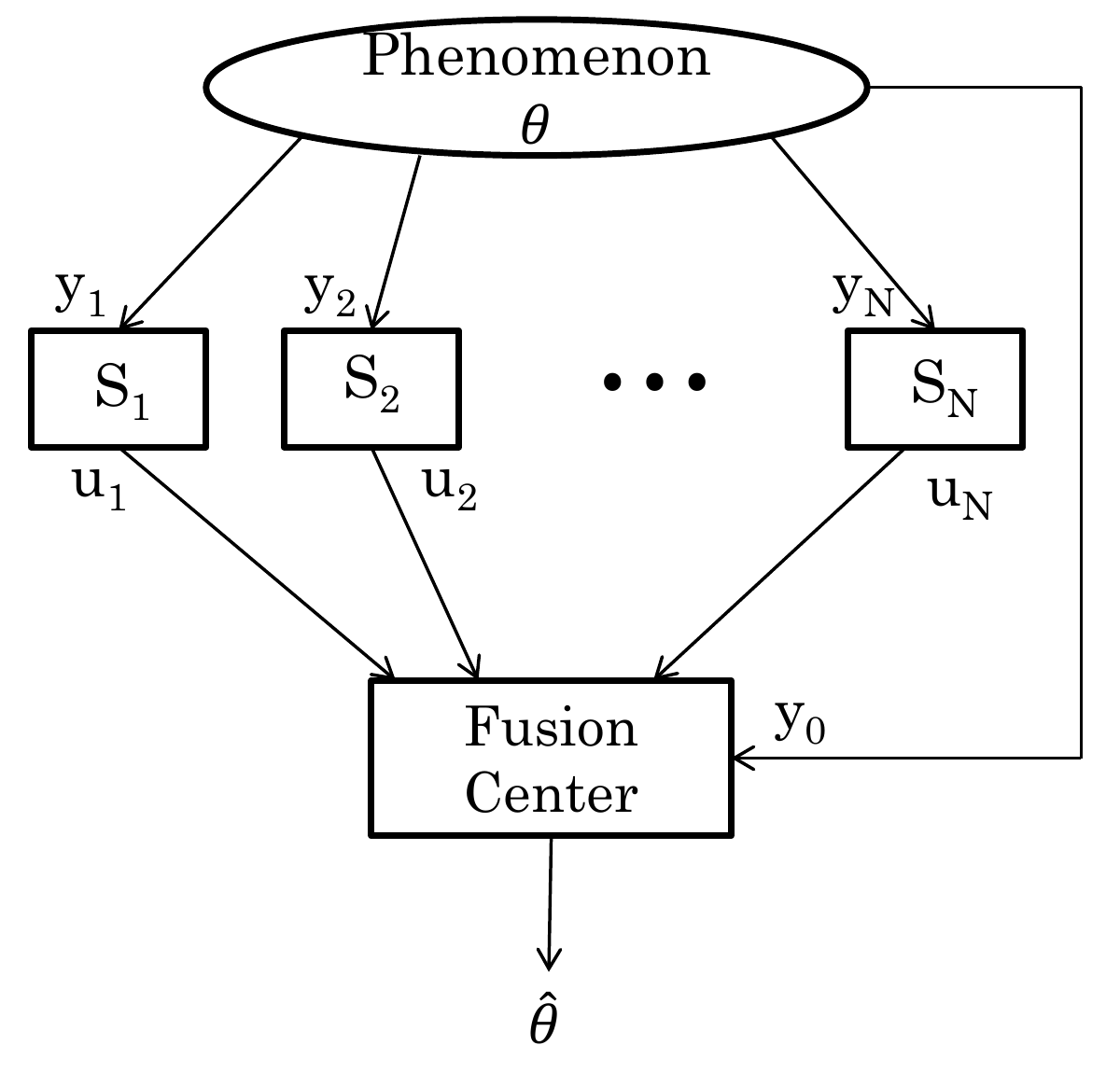}
\caption{System model}
\label{fig:model}
\end{figure}

Each sensor $S_i$, $i \neq 0$, quantizes its observation $y_i$, which is a realization of the random variable $Y_i$, using a local quantizer $\gamma_i(\cdot)$. The quantizer output $u_i = \gamma_i(y_i) \in \{1, \cdots,D_i\}$ is transmitted to the FC error free. Let $L$ denote the number of distinct values of $D_i$ (the number of quantization regions for sensor $S_i$). The FC uses $u_1,\cdots, u_N$ along with its own observation $y_0$ (realization of $Y_0$) and estimates the random parameter $\theta$ as $\hat{\theta} = \gamma_0(y_0, u_1,\cdots, u_N) \in \Theta$. Here $\gamma_0:\mathcal{Y}_0\times \prod_{i}\{1,\cdots,D_i\} \to \Theta$ is a function that will be referred to as the estimator. For $i = 1, 2,\cdots,N$, we use $\Gamma_i$ to denote
the set of all possible quantizers of sensor $S_i$. The collection $\gamma = (\gamma_1,\gamma_2,\cdots ,\gamma_N)$ of quantizers will be referred to as a strategy. The estimator is assumed to be given and, therefore, the strategy only involves local quantizers. We let $\Gamma = \Gamma_1\times\Gamma_1\times\cdots \Gamma_N$, which is the set of all strategies. For $i \neq 0$, once a quantizer  $\gamma_i \in \Gamma_i$ is fixed, the quantizer output $u_i$ at sensor $S_i$ can be viewed as a realization of a random variable $U_i$ defined by $U_i = \gamma_i(Y_i)$. Clearly, the probability distribution of $U_i$ depends on the distribution of $Y_i$ and on the choice of the quantizer $\gamma_i$. Similarly, once the estimator and the strategy are fixed, the global estimate $\hat{\theta}$ becomes a random variable defined by $\hat{\theta} =\gamma_0(Y_0, U_1,\cdots, U_N)$.

In the most general Bayesian formulation, we define a cost function $C : \Theta\times \prod_i\{1, \cdots,D_i\} \times \Theta \to \mathcal{R}$, with $C(\hat{\theta}, u_1, \cdots, u_N, \theta)$ representing the cost associated with an FC estimate $\hat{\theta}$ and quantizer outputs $u_1, \cdots, u_N$, when the true parameter is $\theta$. For any given strategy $\gamma \in \Gamma$, its Bayesian cost (or risk) $J(\gamma)$ is defined as 

\begin{equation}
J(\gamma) = E[C(\hat{\theta}, U_1,\cdots, U_N, \theta)],
\end{equation}

where the arguments of $C(\cdot)$ are all random variables. An equivalent expression for $J(\gamma)$, in which the dependence on $\gamma$ is more explicit is

\begin{align}
\label{cost}
J(\gamma)=\int_\Theta p(\theta)E[C(\gamma_0(Y_0,\gamma_1(Y_1),\cdots,\gamma_N(Y_N)),\gamma_1(Y_1),\cdots,\gamma_N(Y_N),\theta)|\theta]d\theta
\end{align}

\qquad The optimal quantizers are those which minimize $J(\gamma)$, herein referred to as the Bayesian risk function. For a given $\gamma_0(\cdot)$, the problem can be stated as,
\vspace{-.27cm}
\begin{equation}
\label{opt}
\gamma*=\argmin_{\gamma \in \Gamma} J(\gamma).
\end{equation}

\section{Optimality conditions for conditionally independent observations}
\label{sec:opt_cond_iid}
In this section, we provide optimality conditions for quantizers for an arbitrary cost function under the assumption of conditionally independent observations. We first provide a proposition in Sec. \ref{prel} which will be used for deriving the optimality conditions. The results in this section are derived using an approach similar to \cite{Tsitsiklis:bookchapter}.

\subsection{Preliminaries}
\label{prel}
Let $\theta$ be a random parameter to be estimated with prior pdf $p(\theta)$ and $X$ be a random variable, taking values in a set $\mathcal{X}$, with known conditional distribution given $\theta$. Let $D$ be some positive integer, and $\Delta$ be the set of all functions $\delta : \mathcal{X} \to \{1, \cdots ,D\}$. Consistent with our earlier terminology, we shall call such functions quantizers.

\begin{proposition}
\label{prop1}
Let $Z$ be a random variable taking values in a set $\mathcal{Z}$ and assume that, conditioned on $\theta$, $Z$ is independent of $X$. Let $F : \{1, \cdots ,D\} \times \mathcal{Z} \times \Theta \to \mathcal{R}$ be a given cost function. Let $\delta^*$ be an element of $\Delta$. Then $\delta^*$ minimizes $E[F(\delta(X), Z, \theta)]$ over all $\delta \in \Delta$ if and only if
\begin{equation}
\label{prop1_opt}
\delta^*(X)=\argmin_{d=1,\cdots, D}\int_\theta a(\theta ,d)p(\theta|X)d\theta \qquad \text{with probability 1}
\end{equation}
where 
\begin{equation}
a(\theta ,d)=E[F(d,Z,\theta)|\theta] \qquad \text{$\forall$ $\theta$, $d$.}
\end{equation} 
\end{proposition}

\begin{proof}
The minimization of $E[F(\delta(X), Z, \theta)]$ over all $\delta \in \Delta$ can be achieved by fixing a value of $X$ and minimizing the expression $E[F(d, Z, \theta)|X]$, over all $d \in \{1,\cdots,D\}$. In other words, the mapping $\delta(X)$ can be determined for every fixed value of $X$. Therefore, it is equivalent to requiring that $\delta(X)$ minimizes $E[F (d, Z, \theta)|X]$, over all $d \in \{1,\cdots , D\}$, with probability 1. The expression being minimized can be re-written as 
$$E[E[F(d,Z,\theta)|\theta,X]|X]$$

which by conditional independence of $X$ and $Z$, is equal to
\begin{equation}
E[E[F(d,Z,\theta)|\theta,X]|X]=\int_\theta E[F(d,Z,\theta)|\theta]p(\theta|X)d\theta.
\end{equation} 
Therefore, conditional independence decouples the design of $\delta^*(X)$ from $Z$, i.e., $\delta^*(X)$ depends on $Z$ only through $a(\theta,d)$.
\end{proof}

We now use the above result to derive the optimality conditions for the quantizers. 

\subsection{Optimality conditions}
\label{cond}
The following proposition gives the necessary conditions for the optimal strategy that minimizes the Bayesian risk $J(\gamma)$ given in \eqref{cost}. 
\begin{proposition}
\label{prop2}
For $i\neq0$, suppose that $\gamma_j \in \Gamma_j$ has been fixed for all $j \neq i$. Then $\gamma_i$ minimizes $J(\gamma)$ over the set $\Gamma_i$ only if
\begin{equation}
\label{prop2_opt}
\gamma_i(Y_i)=\argmin_{d=1,\cdots,D_i}\int_\theta a(\theta,d)p(\theta|Y_i)d\theta \qquad \text{with probability 1}
\end{equation}
where for any $\theta$ and $d$,
\begin{align}
a(\theta,d)=E[C(U_0,U_1,\cdots,U_{i-1},d,U_{i+1},\cdots,U_N,\theta)|\theta] ,
\end{align}
and where each $U_i$, $i \neq 0$ is a random variable defined by $U_i = \gamma_i(Y_i)$ and $U_0 = \gamma_0(Y_0, U_1, \cdots , U_{i-1}, d, U_{i+1},\cdots,U_N)$.
\end{proposition}

\begin{proof}
Observe that the function to be minimized is 
$$E[C(U_0, U_1, \cdots , U_{i-1}, \gamma_i(Y_i), U_{i+1}, \cdots , U_N,\theta)|\theta]$$
over $\gamma_i \in \Gamma_i$ where $U_0 = \gamma_0(Y_0, U_1,\cdots, U_{i-1},\gamma_i(Y_i),U_{i+1}, \cdots , U_N)$. This is of the form considered in Proposition~\ref{prop1} where $X=Y_i$, $d =\gamma_i(X) = \gamma_i(Y_i)$, $Z$ is the random vector given by $Z = (Y_0, U_1,\cdots, U_{i-1}, U_{i+1},\cdots, U_N)$ and
$F(d, Z, \theta) = C(U_0, U_1, \cdots , U_{i-1},\gamma_i(Y_i), U_{i+1}, \cdots , U_N, \theta)$. The result follows from Proposition~\ref{prop1}. 
\end{proof}

Simultaneously solving $N$ nonlinear equations, i.e., \eqref{prop2_opt} for $i=1,\cdots, N$, is prohibitively challenging. Thus, typically person-by-person optimization (PBPO) approach is used where each decision rule is optimized while decision rules at all other sensors remain fixed. Convergence, at least to a local optimal point, is guaranteed for this greedy approach. For the remainder of the paper, we consider the design of optimal quantizers for a specific cost function namely the Mean-Square Error (MSE). In other words, 
$$C(\hat{\theta}, U_1,\cdots, U_N, \theta) = E[(\hat{\theta} - \theta)^2],$$
where $\hat{\theta} =\gamma_0(Y_0, \gamma_1(Y_1),\cdots, \gamma_N(Y_N))$ and $\theta$ is the true parameter value.

\section{Quantizers for conditionally unbiased and efficient estimators}
\label{sec:unbiased}
In this section, we find the optimal quantizers in distributed estimation for estimators which are efficient and conditionally unbiased. By conditionally unbiased, we mean $E_{x|\theta}[\hat{\theta}] = \theta$ for all $\theta$. The motivation behind such an analysis is that most of the widely used estimators, among them maximum likelihood estimator and maximum a posteriori estimator, are asymptotically unbiased and efficient. In such a scenario, the cost function (MSE) becomes the variance of the estimator which attains the Posterior Cram\'{e}r-Rao Lower Bound (PCRLB). Therefore, the optimization problem can now be formulated as the minimization of PCRLB, or equivalently, the maximization of posterior Fisher Information. Since $\gamma_0(\cdot)$ is assumed to be a fixed efficient, conditionally unbiased estimator, the optimization is now performed over $\gamma = (\gamma_1, \gamma_2, \cdots , \gamma_N)$. While our results hold for any estimator that achieves the PCRLB, the design methodology also applies to cases where no efficient estimator exists; the optimization is therein on performance bounds that are not necessarily attainable but serve as surrogates for estimator performance.
\begin{proposition}
\label{prop3} 
Let $\Gamma$ denote the set of all possible strategies for the distributed estimation problem with identical and conditionally independent distributed sensor observations and $\Gamma^I$ denote the set of all strategies in which all peripheral sensors with the same number of decision regions use identical quantizers. If an efficient and unbiased estimator exists at the Fusion Center, there is no loss in estimation performance in terms of MSE by restricting the search space of optimal strategy to $\Gamma^I$. In other words, if an efficient and conditionally unbiased estimator exists at the FC, there exists an optimal strategy wherein all the peripheral sensors with the same bit-constraint use identical quantization rules.
\end{proposition}

\begin{proof} 
The posterior Fisher Information under the conditional independence assumption is given by
\begin{eqnarray}
F(\gamma)&=& -E_{\theta,\textbf{U},Y_0}[\nabla_\theta\nabla_\theta^T \ln p(\textbf{U},Y_0,\theta)]\\
&=&-E_{\theta,\textbf{U}}[\nabla_\theta\nabla_\theta^T\ln p(\textbf{U}|\theta)] - E_{\theta,Y_0}[\nabla_\theta\nabla_\theta^T\ln p(Y_0|\theta)]-E_\theta[\nabla_\theta\nabla_\theta^T\ln p(\theta)]\\
&=&F_D+F_0+F_P, \label{FI}
\end{eqnarray}
where $F_D$, $F_0$ and $F_P$ represent the local sensor data's contribution, FC data's contribution and
prior's contribution to $F$ respectively.

Since the prior's contribution to $F$ given by $F_P$ and FC's contribution given by $F_0$ are independent of $\gamma$, the optimization problem can be re-stated as
\begin{align}
\label{re-opt}
\gamma^{opt}=\argmax_{\gamma \in \Gamma} F_D=\argmin_{\gamma \in \Gamma} E_{\theta,\textbf{U}}\left[\frac{\partial^2\ln p(\textbf{U}|\theta)}{\partial\theta^2}\right].
\end{align}

As the sensor observations $(Y_1,\cdots, Y_N)$ are conditionally independent and the quantizers $\gamma_i$ are independent of each other, the quantizer outputs are also conditionally independent, i.e. $\ln p(\textbf{U}|\theta) = \sum_{i=1}^N \ln p(U_i|\theta)$.
The objective function now becomes
\begin{equation}
E_{\theta,\textbf{U}}\left[\frac{\partial^2\ln p(\textbf{U}|\theta)}{\partial\theta^2}\right]=\sum_{i=1}^N E_{\theta,U_i}\left[\frac{\partial^2\ln p(U_i|\theta)}{\partial\theta^2}\right].
\end{equation}

The solution to this problem is
\begin{equation}
\label{coupled}
\gamma^{opt}=\argmin_{\gamma \in \Gamma} \sum_{i=1}^N E_{\theta,U_i}\left[\frac{\partial^2\ln p(U_i|\theta)}{\partial\theta^2}\right].
\end{equation}

Observe that \eqref{coupled} can be decoupled into $N$ optimization problems given by
\begin{equation}
\label{decoupled}
\gamma_i^{opt}=\argmin_{\gamma_i \in \Gamma_i} E_{\theta,U_i}\left[\frac{\partial^2\ln p(U_i|\theta)}{\partial\theta^2}\right] \qquad \text{for $i=1,\cdots,N$}.
\end{equation}

Note that, when all the peripheral sensors have identical statistics, the above $N$ optimization problems can be split into $L$ groups, each consisting of identical optimization problems, where $L$ is the number of distinct values of $D_i$. Each of these optimization problems within a group give identical solutions. Therefore, for an efficient and conditionally unbiased estimator, no loss is incurred when only $L$ different quantizers are used at the local sensors. 
\end{proof}

Proposition \ref{prop3} states that, for the special case when all the sensors send the same number of bits, we can constrain all peripheral sensors to use the same quantization rule, without increasing the MSE of the efficient, conditionally unbiased estimator. Furthermore, this optimal quantizer can be found by solving the optimization problem in \eqref{decoupled}. Note that this result holds for any network size $N$. In the next section, we put an additional constraint on the total number of bits that the sensors can transmit to the FC at a given time instant. Under such constraint, we answer the question of whether it is better to have more sensors sending less number of bits/sensor or fewer sensors sending a higher number of bits/sensor? 

\section{Quantizer design under rate constraints}
\label{sec:rate}
In the previous sections, we have found that identical quantizers are optimal when all sensors have the same number of decision regions. The next question to be addressed is the form of the quantizer, or in other words, if each sensor uses a $D_i$ level quantizer, what is the optimal value of $D_i$? The above question has been answered for a detection problem by Chamberland and Veeravalli in \cite{Chamberland&Veeravalli:SP03}. This problem can be solved under a rate constraint on the total number of bits that can be transmitted via the multiple access channel available to the sensors. Consider the scenario where the sensor network is limited by the capacity of the wireless channel over which the local sensors are transmitting their data. Although we consider the presence of a wireless channel between local sensors and fusion center, we consider the simplified case of ideal channels here since the goal of this work is to understand the optimal bit allocation for the sensors\footnote{For a discussion on distributed inference under non-ideal channels, the interested reader is referred to \cite{varshney_spmag06}.}. 

Let each sensor $S_i$ quantize its local observation $Y_i$ using quantizer $\gamma_i(\cdot)$ and transmit the quantized value $u_i=\gamma_i (y_i) \in 1,\dots, D_i$ to the fusion center. When the channel is only able to carry $R$ bits of information per unit time, quantizer design should be considered under the following bit-rate constraint
\begin{equation}
\sum_{i=1}^N \lceil \log_2 {D_i} \rceil \leq R.
\end{equation}
We write $\Gamma(R)$ to denote the set of all admissible  strategies corresponding to a channel with capacity $R$.

Under the conditions discussed in Sec. \ref{sec:unbiased}, when we use conditionally unbiased and efficient estimators, the PCRLB or posterior FI can be considered as the performance metric. Note that, as the sensor observations $Y_i, \dots, Y_N$ are conditionally independent and the quantizers $\gamma_i$ are independent of each other, $F_D$ can be further decomposed as
\begin{eqnarray}
F_D(\gamma) = -\sum_{i=1}^N E_{\theta,U_i}\left[\frac{\partial^2 \ln p(U_i|\theta)}{\partial \theta^2}\right] =\sum_{i=1}^N F_i(\gamma_i),
\end{eqnarray}
where $F_i(\gamma_i)$ is the contribution of sensor $S_i$ to the posterior FI.

For quantization strategy $\gamma_i$, the contribution $F_i(\gamma_i)$ of sensor $S_i$ to $ F(\gamma)$ is bounded above by the Fisher Information $I^*$ contained in one observation $Y$, i.e.,
\begin{equation}
F_i(\gamma_i)  \leq  I^*,
\end{equation}
where $I^*=-E_{\theta,y}\left[\frac{\partial^2 \ln p(y|\theta)}{\partial \theta^2}\right]$. This is because when the observation $y$ is quantized, there is a potential loss in information. By the data processing inequality for Fisher Information \cite{Zamir98}, the amount of Fisher information contained in the observations cannot increase due to quantization. Therefore, the Fisher information corresponding to the quantized observation is no greater than the Fisher information in the actual observation.

We now state the conditions under which binary quantizers ($D_i=2$, for all $i$) are optimal.

\begin{proposition} 
\label{prop:cond}
Suppose there exists a binary quantization function $\hat{\gamma_b} \in \Gamma_b$ ($\Gamma_b$ denotes the set of binary functions on the observation space) such that
\begin{equation}
\label{eq:cond_binary}
F_i(\hat{\gamma_b}) \geq \frac{I^*}{2},
\end{equation}
then having $R$ identical sensors, each sending one bit of information is optimal.
\end{proposition}

\begin{proof}
Let rate $R$ and strategy $\gamma=(\gamma_1, \gamma_2, \dots, \gamma_N) \in \Gamma(R)$ be given. To prove the claim, we construct an admissible binary strategy $\gamma' \in \Gamma(R)$ such that $F(\gamma') \geq F(\gamma)$. First, we divide the collection of decision rules $\gamma=\{\gamma_1, \gamma_2, \dots, \gamma_N\}$ into two sets, the first set contains all the binary functions and the other composed of the remaining quantization rules. We define $S_b$ to be the set of integers for which the function $\gamma_i$ is a binary decision rule
\begin{equation}
S_b=\{i: 1 \leq i \leq N, \gamma_i \in \Gamma_b\}.
\end{equation}
Similarly, we let $S_{b}^c=\{1, 2, \dots, N\}$-$S_b$. We choose a binary decision rule $\hat{\gamma_b} \in \Gamma_b$ such that
\begin{equation}
F_i(\hat{\gamma_b}) \geq \max \{ \max_{i \in S_b} \{F_i(\gamma_i) \}, \frac{I^*}{2} \}.
\end{equation}
Such a function $\hat{\gamma_b}$ always exists by assumption $F_i(\hat{\gamma_b}) \geq \frac{I^*}{2}$. Notice that when $i \in S_{b}^c$, $D_i >2$ which implies that $\lceil \log_2 (D_i) \rceil  \geq 2$. We can replace each sensor with index in $S_{b}^c$ by two binary sensors without exceeding the capacity of the channel. Then we consider the alternative scheme $\gamma'$ in which we replace every sensor with index in $S_b$ by a binary sensor with decision rule $\hat {\gamma_b}$.
\begin{eqnarray}
F(\gamma')&=&(|S_b|+2|S_{b}^c|)F_i(\hat{\gamma_b}) \geq |S_b|F_i(\hat{\gamma_b})+|S_{b}^c|I^* \nonumber\\
&\geq& F(\gamma).
\end{eqnarray}
For a fixed decision rule $\hat {\gamma_b}$, the Fisher information at the fusion center is monotonically increasing in the number of sensors. We can, therefore, improve performance by increasing the number of sensors in $\gamma'$ until the rate constraint $R$ is satisfied with equality. The strategy $\gamma$ being arbitrary, we conclude that having $N=R$ identical sensors, each sending one bit of information, is optimal.
\end{proof}

Combining with the result from Sec. \ref{sec:unbiased}, it is optimal for local sensors to use identical binary quantization rules under certain conditions. We now present an example of the widely used Gaussian observations model to show that the above condition can be achieved for this case.

\subsection*{Gaussian Observations}
\label{sec:Gaussian}
Consider a sensor network consisting of sensors which are estimating a parameter $\theta \sim p_{\Theta}(\cdot)$ with mean $\mu_\theta$ and variance $\sigma^2_{\theta}$. Each sensor $S_i$ receives noisy observations $y_i$ which are governed by Gaussian statistics:

\begin{equation}
\label{eq:Gaussian}
p(y_i|\theta)\sim\mathcal{N}(\theta,\sigma^2),
\end{equation}

where $\mathcal{N}(\theta,\sigma^2)$ denotes a Gaussian distribution with mean $\theta$ and variance $\sigma^2$. In order to check the condition \eqref{eq:cond_binary}, we first evaluate the contribution of a single sensor to the total posterior Fisher information in the Gaussian case.

\begin{lemma}
For observations with Gaussian distributions as in \eqref{eq:Gaussian}, the contribution of a single sensor to the posterior Fisher Information is given by $I^*=\frac{1}{\sigma^2}$.
\end{lemma} 

\begin{proof}
The lemma can be proved by straightforward calculation. Since $p(y|\theta)\sim\mathcal{N}(\theta,\sigma^2)$, we have

\begin{eqnarray}
&&\frac{\partial\ln p(y|\theta)}{\partial \theta}=\frac{(y-\theta)}{\sigma^2}\\
&\implies& \frac{\partial^2 \ln p(y|\theta)}{\partial \theta^2}=-\frac{1}{\sigma^2}
\end{eqnarray}
which gives us the desired result.
\end{proof}

Note that Proposition \ref{prop:cond} states that binary quantizers are optimal if there exists a binary quantizer $\hat{\gamma}_b$ which satisfies the condition \eqref{eq:cond_binary}. Consider the threshold quantizer using threshold $\mu_{\theta}$ as a candidate binary quantizer:

\begin{equation}
\label{threshold_quant_bin}
\hat{\gamma_b}(y)=
\begin{cases}
1, \qquad \text{if $y \geq \mu_\theta$}\\
0, \qquad \text{otherwise}
\end{cases}.
\end{equation}

We find the Fisher information $F_i(\hat{\gamma}_b)$ corresponding to this binary threshold quantizer $\hat{\gamma_b}$.

\begin{lemma}
For sensor $S_i$ using the binary threshold quantizer $\hat{\gamma_b}$, the posterior Fisher information is given by the following:
\begin{equation}
F_i(\hat{\gamma}_b)= \frac{1}{2\pi\sigma^2}E_{\theta}\left[\frac{\exp{\left(-\frac{(\theta-\mu_\theta)^2}{\sigma^2}\right)}}{(1-Q(\frac{\theta-\mu_\theta}{\sigma}))Q(\frac{\theta-\mu_\theta}{\sigma})}\right],
\end{equation}
where 
\begin{equation}
Q(x)=\frac{1}{\sqrt{2\pi}}\int_{x}^\infty\exp{\left(-\frac{t^2}{2}\right)}dt
\end{equation}
is the complementary cumulative distribution function of Gaussian distribution.
\end{lemma}

\begin{proof}
The proof follows from the fact that for binary quantizer, the FI is given by \cite{chen_tsp10} \begin{equation}
\label{FI_one}
I(\theta)=\frac{(g'(\theta))^2}{g(\theta)(1-g(\theta))},
\end{equation}
where $g'(\theta)$ represents the first derivative of $g(\theta)=P(U_i=1|\theta)$ with respect to $\theta$. For Gaussian observations, we have $g(\theta)=Q(\frac{\theta-\mu_\theta}{\sigma})$ and using the definition of posterior Fisher information gives the desired result.
\end{proof}

Using the above lemmas, we can find the sufficient condition for binary quantizers to be optimal. Note that this is only a sufficient condition and is not necessary for the optimality of binary quantizers.

\begin{theorem}
For Gaussian observations under low signal-to-noise ratio (SNR) regime $\left(\frac{\sigma^2_{\theta}}{\sigma^2} \leq 2\ln\frac{4}{\pi}\right)$, it is optimal to have identical quantizers at all $R$ sensors, each sending one-bit of information. 
\end{theorem}

\begin{proof}
A sufficient condition for binary quantizers to be optimal is \eqref{eq:cond_binary} from Proposition \ref{prop:cond}. We determine the condition under which \eqref{eq:cond_binary} is satisfied. We start by using the following result from \cite{Chamberland&Veeravalli:SP03}: For any $x$,
\begin{equation}
Q(x)Q(-x)\leq\frac{1}{4}e^{(-x^2/2)}.
\end{equation}

From this we have the following set of inequalities
\begin{eqnarray}
F_i(\hat{\gamma}_b)&=&\frac{1}{2\pi\sigma^2}E_{\theta}\left[\frac{\exp{\left(-\frac{(\theta-\mu_\theta)^2}{\sigma^2}\right)}}{(1-Q(\frac{\theta-\mu_\theta}{\sigma}))Q(\frac{\theta-\mu_\theta}{\sigma})}\right]\\
&\geq&\frac{1}{2\pi\sigma^2}E_{\theta}\left[\frac{4\exp{\left(-\frac{(\theta-\mu_\theta)^2}{\sigma^2}\right)}}{\exp{\left(-\frac{(\theta-\mu_\theta)^2}{2\sigma^2}\right)}}\right]\\
&=&\frac{2}{\pi\sigma^2}E_{\theta}\left[\exp{\left(-\frac{(\theta-\mu_\theta)^2}{2\sigma^2}\right)}\right]\\
&\geq&\frac{2}{\pi\sigma^2}\exp{\left(-E_{\theta}\left[\frac{(\theta-\mu_\theta)^2}{2\sigma^2}\right]\right)}\label{eq:jensen}\\
&=&\frac{2}{\pi\sigma^2}\exp{\left(-\frac{\sigma_\theta^2}{2\sigma^2}\right)}\\
&\geq&\frac{2}{\pi\sigma^2}\exp{\left(-\frac{2\ln\frac{4}{\pi}}{2}\right)}=\frac{1}{2\sigma^2}\label{eq:cond_use}
\end{eqnarray}
where, for \eqref{eq:jensen}, we have used Jensen's inequality \cite{cover_inftheory} for the convex exponential function and for \eqref{eq:cond_use}, the condition of low SNR regime $\left(\frac{\sigma^2_{\theta}}{\sigma^2} \leq 2\ln\frac{4}{\pi}\right)$ is used.
\end{proof}

The above theorem states that when the local sensor observations have very low SNR, it is optimal to use identical binary quantizers at local sensors. This is intuitively true because when the SNR is low, the observations do not have a lot of information and, therefore, the sensors do not have to waste their resources and send fine-quantized data. However, this result does not imply that binary quantizers are always optimal. For example, when the observations are correlated, binary quantizers need not be optimal. To illustrate how correlation affects our results, we study the specific case of estimation of mean in equicorrelated Gaussian noise. In this case, the observations have the following distribution 

$$p(\mathbf{y}|\theta)\sim\mathcal{N}(\theta\mathbf{1},\mathbf\Sigma),$$
where $\mathbf 1$ is the $N\times 1$ column vector of all ones and $\Sigma$ is the covariance matrix where the diagonal elements are $\sigma^2$ and the off-diagonal elements are $\rho\sigma^2$. Here $\rho$ is the correlation coefficient. The Fisher information obtained from $N$ observations is given by 
$2\mathbf{1}^T\Sigma^{-1}\mathbf{1}$. Fig.~\ref{fig:corr} shows the amount of information contained in observations for unit noise variance, $\sigma^2=1$. As the correlation coefficient goes to one, the amount of information contained in $N$ observations approaches the amount of information contained in one observation. Hence, in the limit, having one sensor sending $R$ bits of information is optimal. This suggests that correlation in the observations favors having fewer sensors sending multiple bits, or having nonidentical sensors, rather than employing a set of identical binary
sensors. Similar observations were also made by Chamberland and Veeravalli in \cite{Chamberland&Veeravalli:SP03} for the case of decentralized detection in sensor networks. 

\begin{figure}[htb]
\centering
\includegraphics[width=3.5in,height=!]{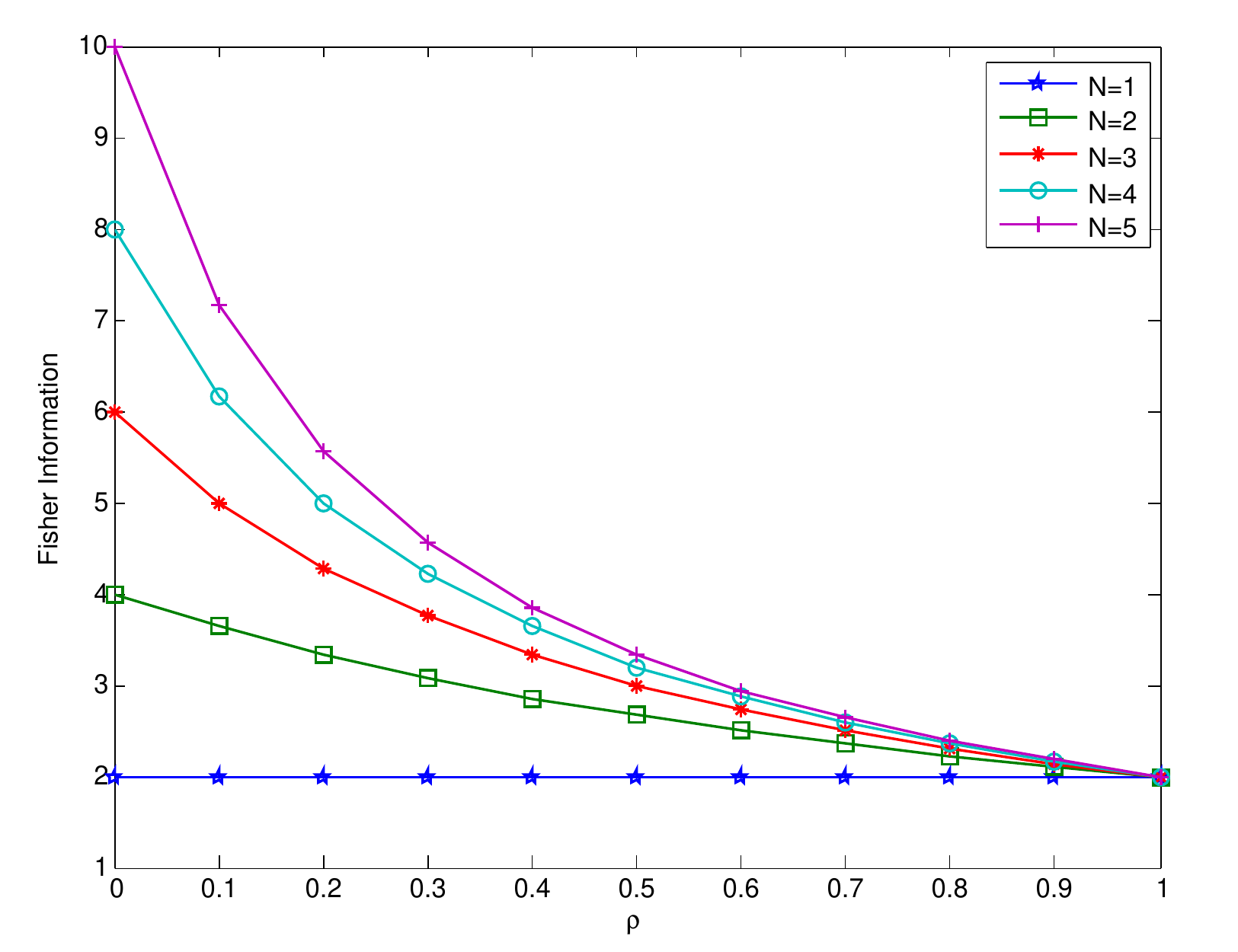}
\caption{Fisher information contained in $N$ sensor observations as a function of correlation coefficient $\rho$.}
\label{fig:corr}
\end{figure}

\section{The Location Parameter Estimation problem}
\label{sec:location}
Having shown that identical binary quantizers are optimal under certain conditions, in this section, we present a methodology to design this optimal identical binary quantizer for a location parameter estimation problem. Consider the location parameter estimation problem where the observations are corrupted by independent and identically distributed (i.i.d) additive noise with pdf $p_W(w)$.
\begin{equation}
\label{AN}
Y_i = \theta +W_i \qquad \text{for $i = 1,\cdots,N$,}
\end{equation}
where $\theta \sim p_\Theta(\theta)$ and $W_i$ is the i.i.d noise. The local sensors process their own observations locally before sending the processed data ($U_i$ for $i = 1,\cdots,N$) to the FC. The FC then estimates $\theta$ from $\textbf{U} = [U_1 \cdots U_N]$ and $Y_0$. We consider the problem of designing the binary quantizers which we have shown to be optimal under certain conditions. Also, as shown in the previous sections, for an efficient and unbiased estimator at the FC, the identical quantizers are optimal. Let the quantizer be represented by $\gamma(Y)$ which maps the data $Y_i$ to one of the two bit values $\{0, 1\}$. We represent the quantizers probabilistically as
\begin{equation}
\gamma(Y_i)=P(U_i=1|Y_i).
\end{equation}

Thus, $\gamma(Y_i)$ denotes the probability with which the $i^{th}$ local sensor sends a `1' to the FC given its observation, $Y_i$. Stochastic quantizers are employed here as they cover a wide range of possible quantizers including both the threshold quantizers and the dithering quantizers.

\subsection{Posterior Cram\'{e}r Rao lower bound}
For the location parameter estimation problem, $F$ from \eqref{FI} is the posterior Fisher Information \cite{vantrees_bounds} which is a function of the prior distribution $p_\Theta(\theta)$, the quantizer $\gamma(Y)$, and the noise pdf $p_W(w)$. It is given as 
\begin{equation}
\label{total_FI}
F(p_\Theta,\gamma, p_W) = F_D + F_0 + F_P,
\end{equation}
where $F_D$, $F_0$ and $F_P$ are as defined before.

As the $N$ observations are conditionally independent, we have 
\begin{equation}
\label{FD}
F_D = NE_\theta[I(\theta)],
\end{equation}
where $I(\theta)$ is given by \eqref{FI_one}.

Let $g(\theta)$ denote the probability that the quantizer output is `1' given the true value of $\theta$
\begin{align}
g(\theta)=P(U_i=1|\theta)&=E_{W_i}[\gamma(\theta +W_i)]\\
&=\int_y \gamma(y)p_W(y-\theta)dy.\label{g}
\end{align}

For a binary quantizer, the FI is given by \eqref{FI_one}. From \eqref{total_FI}, \eqref{FD} and \eqref{FI_one}, the posterior FI is given by
\begin{equation}
\label{FI_total}
F(p_\Theta,\gamma,p_W)=N\int_\theta \frac{(g'(\theta))^2}{g(\theta)(1-g(\theta))}p_\Theta(\theta)d\theta + F_0 +F_P.
\end{equation}

\subsection{Optimal quantizer design}
The optimal quantizer $\gamma^*(y)$ minimizes the PCRLB or, equivalently, maximizes $F(p_\Theta,\gamma, p_W)$. Since $F_0$ and $F_P$ are independent of the quantizer, the optimization
problem can be stated as
\begin{align}
\gamma^*(\cdot)=\argmax_{\gamma(\cdot)}F_D=\argmax_{\gamma(\cdot)}\int_\theta \frac{(g'(\theta))^2}{g(\theta)(1-g(\theta))}p_\Theta(\theta)d\theta.
\label{opt_FD}
\end{align}

This problem can be solved by observing that the objective function $F_D$ depends on $\gamma(\cdot)$ only
through $g(\theta)$ given in \eqref{g} which can be re-written as 
$$g(\theta)=(\gamma(y)*p_W(-y))(\theta),$$ 
where `*' represents the convolution operation. Transforming this into frequency domain using the Fourier Transform, we get 
$$G(f) = H(f)P_W(-f),$$ 
where $G(f)$, $H(f)$ and $P_W(f)$ are the Fourier transforms of $g(\cdot),\gamma(\cdot)$ and $p_W(\cdot)$ respectively. Therefore, given the noise pdf $p_W(\cdot)$, the quantizer $\gamma(\cdot)$ can be found (if it exists) as
\begin{equation}
\label{inv_FT}
\gamma(y)=F^{-1}\left[\frac{G(f)}{P_W(-f)}\right],
\end{equation}
where $F^{-1}$ is the Inverse Fourier transform.

The problem now reduces to that of finding the optimal $g(\theta)$ that maximizes the integrand in \eqref{opt_FD}. Note that this optimal $g^*(\theta)$ is independent of the noise pdf $p_W(w)$. Upon obtaining $g^*(\theta)$, the optimal quantizer for a given noise pdf can then be designed using \eqref{inv_FT}. Therefore, the optimization in \eqref{opt_FD} can be re-stated as
\begin{eqnarray}
g^*(\cdot)&=&\argmax_{g(\cdot)}L(g(\cdot))\\
&=&\argmax_{g(\cdot)}\int_\theta I(\theta)p_\Theta(\theta) d\theta,\label{opt_FD_g}
\end{eqnarray}
where $L(g(\cdot))=\int_\theta I(\theta)p_\Theta(\theta) d\theta$ and $I(\theta)$ is given in \eqref{FI_one}.

\begin{proposition}
\label{prop4}
Given the prior distribution $p_\Theta(\theta)$, the optimal $g^*(\theta)$ can be found by solving the following differential equation\footnote{Note that this gives a stationary point which needs to be verified to be a maximum.}
\begin{equation}
\label{diff_eq}
p_\Theta(\theta)(g'(\theta))^2(1-2g(\theta)) = 2g(\theta)(1-g(\theta))(g''(\theta)p_\Theta(\theta)+g'(\theta)p'_\Theta(\theta)),
\end{equation}
where $'$ and $''$ denote respectively the first and the second derivatives with respect to $\theta$.
\end{proposition}

\begin{proof}
Define $K(\theta) = I(\theta)p_\Theta(\theta)$ as the function of $\theta$ which is the integrand in \eqref{opt_FD_g}. The optimization problem presented in \eqref{opt_FD_g} is a typical variational calculus problem and it can be solved using the Euler-Lagrange equation \cite{brunt04} stated below
\begin{equation}
\label{EL}
\frac{\partial K}{\partial g}=\frac{d}{d\theta}\frac{\partial K}{\partial g'}.
\end{equation}
From the expression of $I(\theta)$ given in \eqref{FI_one}, we have
\begin{equation}
\label{LHS}
\frac{\partial K}{\partial g}=-\frac{(g')^2p_\Theta (1-2g)}{(g-g^2)^2}
\end{equation}
and
\begin{equation}
\label{RHS}
\frac{\partial K}{\partial g'}=\frac{2g'p_\Theta}{(g-g^2)}.
\end{equation}

Differentiating \eqref{RHS} with respect to $\theta$ and using \eqref{EL}, we get the desired result. 
\end{proof}

As can be seen from \eqref{diff_eq}, the differential equation can be solved for a given prior $p_\Theta(\theta)$. After finding this optimal $g^*(\theta)$, the optimal quantizer $\gamma^*(x)$ can be found for a given noise pdf $p_W(w)$ using \eqref{inv_FT}.

\subsection{Example: Least favorable prior}
In this section, we consider a special case of $\theta$ following the least favorable prior and find the optimal
$g^*(\theta)$. Note that when we have a least favorable prior, the Bayesian criterion matches with the minimax criterion. Therefore, the optimal quantizer design is now the following:

\begin{equation}
\label{eq:minimax}
g^*(\theta)=\argmax_{g(\cdot)} \min_{\theta} I(\theta)
\end{equation}

\begin{proposition}
\label{prop}
Given that $\theta$ follows least favorable prior with support $[\theta_{min}, \theta_{max}]$, the solution to the optimization problem in \eqref{eq:minimax}, $g^*(\theta)$ is given by
\begin{equation}
\label{gen_sin}
g^*(\theta)=\frac{1}{2}\left[1+\sin{\pi\left(\frac{\theta-\theta_{min}}{\theta_{max}-\theta_{min}}-\frac{1}{2}\right)}\right], \text{ $\theta \in [\theta_{min},\theta_{max}]$.}
\end{equation}
\end{proposition}
\begin{proof} 
Note that the minimax solution to \eqref{eq:minimax} is the one where the function $I(\theta)$ is a constant.
Therefore, 
\begin{eqnarray}
&&I(\theta) = c^2\\
\implies&&\frac{(\frac{dg}{d\theta})^2}{g(\theta)(1-g(\theta))}=c^2\\
\implies&&\frac{dg}{\sqrt{g(1-g)}}=cd\theta,
\end{eqnarray}
where $c$ is a constant. Without loss of generality, assuming the boundary conditions as $g(\theta_{min})=0$ and $g(\theta_{max})=1$, we obtain the result as $g^*(\theta)$ as

\begin{equation}
g^*(\theta)=\frac{1}{2}\left[1+\sin{\pi\left(\frac{\theta-\theta_{min}}{\theta_{max}-\theta_{min}}-\frac{1}{2}\right)}\right], \text{ $\theta \in [\theta_{min},\theta_{max}]$.}
\end{equation}
\end{proof}

Note that the same result was obtained by Chen and Varshney \cite{chen_tsp10} when directly using the minimax CRLB as the performance metric for a distributed estimation problem with deterministic unknown parameter $\theta$. 

Without loss of generality, let $\theta_{min} = -1$ and $\theta_{max} = 1$. The optimal $g^*(\theta)$ given in \eqref{gen_sin} becomes
\begin{equation}
\label{canonical}
g^*(\theta)=\frac{1}{2}\left[1+\sin{\frac{\pi\theta}{2}}\right], \qquad \text{ for $\theta \in [-1,1]$}
\end{equation}

\subsubsection{Noiseless observations}
\label{lim}
The performance limit of this distributed estimation problem under the least favorable Bayesian criterion can be characterized by observing the performance when the observations are noiseless. When these observations at the local sensors prior to quantization are noiseless, i.e., the observation model is perfect, $p_W(w) = \delta(w)$. The optimal quantizer, for this case, is given by the sine quantizer
\begin{equation}
\label{sine_quant}
\gamma^*(y)=\frac{1}{2}\left[1+\sin{\frac{\pi y}{2}}\right], \qquad \text{ for $y \in [-1,1]$}
\end{equation}

In this case, the Fisher information is $F = \frac{N\pi^2}{4}$ and the CRLB is $\frac{4}{N\pi^2}$, where $N$ is the total number of sensors. This represents the performance limit under the Bayesian criteria for the distributed location parameter estimation problem with least favorable prior.

\subsubsection{Optimality of threshold quantizers}
Threshold quantizers are the most widely used quantizers due to their simplicity \cite{gian_tsp06_1}. A threshold quantizer is given by

\begin{equation}
\label{threshold_quant}
\gamma_T(y)=
\begin{cases}
1, \qquad \text{if $y \geq T$}\\
0, \qquad \text{otherwise}
\end{cases}.
\end{equation}

An interesting question is to find the conditions on the noise pdf $p_W(w)$ for which the threshold quantizers attain the performance limit as described in Sec. \ref{lim} which is the performance when the
observations are noiseless (refer to the discussion after \eqref{sine_quant}). For the optimality condition
to be satisfied, the threshold quantizer and the noise distribution should satisfy the following constraint
\begin{align}
g^*(\theta)&=\int_y\gamma_T(y)p_W(y-\theta)dy\\
&=\int_{y=T}^\infty p_W(y-\theta)dy=1-F_W(T-\theta),
\end{align}
where $F_W(w)$ is the cumulative distribution function of noise and $g^*(\theta)$ is given by \eqref{sine_quant}. Differentiating both sides and using the fact $\frac{dF_W(w)}{dw} = p_W(w)$, we get the sufficient condition for the threshold quantizer $\gamma_T(y)$ to achieve performance limit when the noise pdf is
\begin{align}
\label{thresh_cond}
p_W(w)=
\begin{cases}
\frac{\pi}{4}\cos{\frac{\pi}{2}(w-T)},	\qquad &\text{for $w \in [T-1, T+1]$}\\
0,		\qquad &\text{otherwise}
\end{cases}.
\end{align}
Threshold quantizers can still be optimal for a wide range of noise distributions (as shown in \cite{kar_tsp12} for minimax CRLB criterion) but the performance limit can be reached only for the above noise pdf. We now show via simulations that when the observations are corrupted by the above noise pdf, using threshold quantizers allows us to achieve the performance limit when the estimator is conditionally unbiased and efficient.

\paragraph*{Simulation Results}
For the sake of tractability, we consider the maximum likelihood estimator (MLE) at the FC, which is asymptotically conditionally unbiased and efficient. Therefore, as $N\to\infty$, the MSE of the estimate should attain the performance limit. The MLE of $g(\theta)$ is given by 
\begin{equation}
\hat{g}(\theta)=\frac{\sum_{i=1}^Nu_i}{N}.
\end{equation}

By invariance property of MLE \cite{kay_est}, we get the ML estimate of $\theta$, as

\begin{equation}
\hat{\theta}=\frac{2}{\pi}\sin^{-1}\left(2\frac{\sum_{i=1}^Nu_i}{N}-1\right).
\label{estimate}
\end{equation}

Consider noisy observations of the location parameter corrupted by additive noise with distribution given in \eqref{thresh_cond} with $T=0$. The local sensors quantize their observations using the threshold quantizer with threshold $T=0$. The FC uses the estimator $\hat{\theta}$ of \eqref{estimate} to estimate the unknown parameter $\theta$. In Fig. \ref{fig:performance}, we plot the MSE of $N_{mc}=5000$ Monte-Carlo runs as a function of the number of sensors. As the figure shows, the MSE reaches the performance limit as $N\to\infty$. This is expected since the estimator at the FC, ML estimator, is asymptotically unbiased and efficient. Therefore, threshold quantizer is asymptotically optimal among all quantizers. 

\begin{figure}[htb]
\centering
\includegraphics[width=3.5in,height=!]{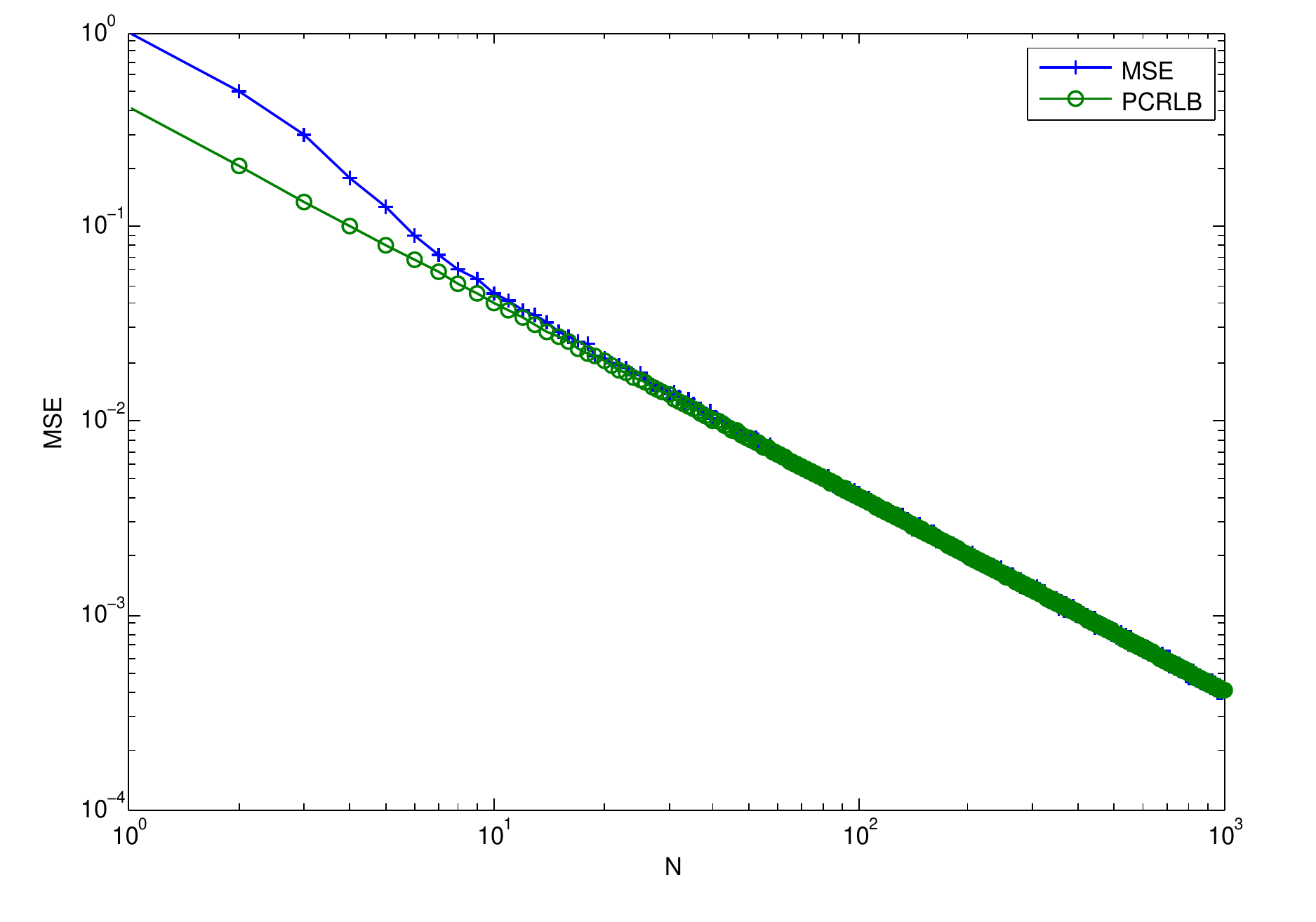}
\caption{MSE of ML estimator approaching the performance limit (PCRLB or minimax CRLB) as $N\to\infty$}
\label{fig:performance}
\end{figure}

\section{Optimality conditions for conditionally dependent observations}
\label{sec:dep}
In this section, we relax the assumption of conditionally independent observations and address the quantizer design problem when the observations are conditionally dependent across sensors. For convenience, we consider the case when the FC does not make any observations of its own and, therefore, the observations are $\textbf{Y}=[Y_1,\cdots, Y_N]$. We derive the optimality conditions by considering the hierarchical conditional independence (HCI) model proposed by Chen~\textit{et.~al.} in \cite{Chen_dep}. This framework introduces a hidden random variable which simplifies the analysis of the system. Consider the distributed estimation system shown in Fig. \ref{fig:model}. When the FC does not make its own observation, the system follows the following Markov Chain:

\begin{equation}
\theta \to \textbf{Y} \to \textbf{U} \to \hat{\theta}.
\end{equation}

Now when the observations are conditionally independent, the conditional distribution factorizes as $p(\textbf{y}|\theta)=\prod_{i=1}^N p(y_i|\theta)$. However, when the observations are not conditionally independent, we cannot factorize the conditional distribution of the observations. Instead, the proposed HCI framework introduces a new hidden random variable $\lambda$ such that the following Markov chain holds:

\begin{equation}
\theta \to \lambda \to \textbf{Y} \to \textbf{U} \to \hat{\theta}
\end{equation}

and the observations are conditionally independent given this hidden random variable $\lambda$. In other words, 

\begin{equation}
p(\textbf{y}|\lambda)=\prod_{i=1}^N p(y_i|\lambda)
\end{equation}
even if $p(\textbf{y}|\theta)\neq\prod_{i=1}^N p(y_i|\theta)$. The equivalence between any general distributed inference model and the HCI model has been discussed in \cite{Chen_dep}. Under this framework, we now derive the optimality conditions of the quantizer for any cost function $C(\hat{\theta},u_1,\cdots,u_N,\theta)$. We first provide a proposition which will be used for deriving the optimality conditions. The results in this section are derived in a manner similar to Sec. \ref{sec:opt_cond_iid}.

Let $\theta$ be a random parameter to be estimated with prior pdf $p(\theta)$ and let $X$ be a random variable, taking values in a set $\mathcal{X}$, with known conditional distribution given $\theta$. Let $D$ be some positive integer, and let $\Delta$ the set of all functions $\delta : \mathcal{X} \to \{1, \cdots ,D\}$. Consistent with our earlier terminology, we shall call such functions quantizers.

\begin{proposition}
\label{prop1_dep}
Let $Z$ be a random variable taking values in a set $\mathcal{Z}$ and assume that, conditioned on $\lambda$, $Z$ is independent of $X$. Let $F : \{1, \cdots ,D\} \times \mathcal{Z} \times \Theta \to \mathcal{R}$ be a given cost function. Let $\delta^*$ be an element of $\Delta$. Then $\delta^*$ minimizes $E[F(\delta(X), Z, \theta)]$ over all $\delta \in \Delta$ if and only if
\begin{eqnarray}
\label{prop1_opt}
\delta^*(X)=\argmin_{d=1,\cdots, D}\int_{\theta,\lambda} a(\theta ,\lambda, d)p(\theta,\lambda|X)d\theta d\lambda \\
\text{with probability 1}\nn
\end{eqnarray}
where 
\begin{equation}
a(\theta, \lambda, d)=E[F(d,Z,\theta)|\theta,\lambda] \qquad \text{$\forall$ $\theta$, $\lambda$, $d$.}
\end{equation} 
\end{proposition}

\begin{proof}
The proof is similar to the proof of Proposition \ref{prop1} with the inclusion of $\lambda$. The minimization of $E[F(\delta(X), Z, \theta)]$ over all $\delta \in \Delta$ is equivalent to requiring that $\delta(X)$ minimize $E[F(d, Z, \theta)|X]$, over all $d \in \{1,\cdots,D\}$, with probability 1. The expression being minimized can be re-written as 
$$E[E[F(d,Z,\theta)|\theta,\lambda,X]|X]$$

which by conditional independence of $X$ and $Z$ given $\lambda$, is equal to

\begin{equation}
E[E[F(d,Z,\theta)|\theta,\lambda,X]|X]=\int_{\theta,\lambda} E[F(d,Z,\theta)|\theta,\lambda]p(\theta,\lambda|X)d\theta d\lambda
\end{equation} 
Therefore, conditional independence decouples the design of $\delta^*(X)$ from $Z$, i.e., $\delta^*(X)$ depends on $Z$ only through $a(\theta,\lambda)$.
\end{proof}

We now use the above result to derive the optimality conditions for the quantizers when observations are conditionally dependent. 

\begin{proposition}
\label{prop2_dep}
Fix $i$ and suppose that $\gamma_j \in \Gamma_j$ has been fixed for all $j \neq i$. Then $\gamma_i$ minimizes $J(\gamma)$ over the set $\Gamma_i$ only if
\begin{eqnarray}
\label{prop2_opt_dep}
\gamma_i(Y_i)=\argmin_{d=1,\cdots,D_i}\int_{\theta,\lambda} a(\theta,\lambda,d)p(\theta,\lambda|Y_i)d\theta d\lambda \\
\text{with probability 1,}\nn
\end{eqnarray}
where for any $\theta$, $\lambda$, and $d$,
\begin{align}
a(\theta,\lambda,d)=E[C(U_0,U_1,\cdots,U_{i-1},d,U_{i+1},\cdots,U_N,\theta)|\theta,\lambda] 
\end{align}
and where each $U_i$, $i \neq 0$ is a random variable defined by $U_i = \gamma_i(Y_i)$ and $U_0 = \gamma_0(U_1, \cdots , U_{i-1}, d, U_{i+1},\cdots,U_N)$.
\end{proposition}

\begin{proof}
Observe that the minimization is of 
$$E[C(U_0, U_1, \cdots , U_{i-1}, \gamma_i(Y_i), U_{i+1}, \cdots , U_N,\theta)|\theta],$$
over $\gamma_i \in \Gamma_i$ where $U_0 = \gamma_0(U_1,\cdots, U_{i-1},\gamma_i(Y_i),U_{i+1}, \cdots , U_N)$. This is of the form considered in Proposition~\ref{prop1_dep} where $X=Y_i$, $d =\gamma_i(X) = \gamma_i(Y_i)$, $Z$ is the random vector given by $Z = (U_1,\cdots, U_{i-1}, U_{i+1},\cdots, U_N)$ and
$F(d, Z, \theta) = C(U_0, U_1, \cdots , U_{i-1},\gamma_i(Y_i), U_{i+1}, \cdots , U_N, \theta)$. The result follows from Proposition~\ref{prop1_dep}. 
\end{proof}

Proposition \ref{prop2_dep} is similar to Proposition \ref{prop2} and provides the necessary conditions for optimal quantizers for an arbitrary cost function $C(\cdot)$. We would like to note that the other results derived in the case of conditionally independent observations may not always be true when the observations are dependent. For example, when the observations are dependent, it can be easily seen that identical quantizers are not optimal in general. Consider the following simple example: there are $N=2^n-1$ sensors in the network which send binary quantized version of their observations to the FC. The local sensor observation model is given as follows:
$$y_i=\theta+v_i$$
where $\theta\in[-1,1]$ and $v_i=\delta(v-v_0)$ for all $i$. In other words, the single-peak noise is perfectly correlated across sensors. When all sensors use an identical quantizer $\gamma$, the quantized observation received from every sensor is the same (say all 1). On the other hand, we can easily design non-identical quantizers which provide additional information as follows: split the region $[-1,1]$ into $2^n$ equal regions, and the sensor $i$ uses a threshold quantizer to test whether $\theta$ lies in the first $i$ regions or not. In this way, we can determine the exact region among the $2^n$ regions where $\theta$ lies. Therefore, identical quantizers are not optimal in this example when observations are dependent. We have also shown in Sec.~\ref{sec:rate} that binary quantizers are not optimal when observations are correlated. Study on the optimal quantizer design for dependent observations will be considered in our future work.

\section{Conclusion}
\label{sec:conc}
In this work, we considered the problem of quantizer design for distributed estimation under the Bayesian criterion.  We showed that for conditionally unbiased efficient estimators, when all the sensors have the same number of decision regions, identical quantizers are optimal. Considering a communication rate constraint on the network, we derived the conditions for the optimality of binary quantizers. We have shown that when the observations are Gaussian, identical binary quantizers are optimal in the low SNR regime. For the location parameter estimation problem with a given prior distribution, we have found the optimal binary quantizer by solving a differential equation. We have found the sufficient condition on the noise distribution for which the threshold quantizers attain the performance limit. By relaxing the assumption of conditionally independent observations at the sensors, we also derived the optimality conditions for quantizers with conditionally dependent observations. In the future, we will further study the open problem of quantizer design in a distributed estimation framework with dependent observations.

\bibliographystyle{IEEEtran}
\bibliography{Conf,Book,Journal}

\begin{thebibliography}{10}
\providecommand{\url}[1]{#1}
\csname url@samestyle\endcsname
\providecommand{\newblock}{\relax}
\providecommand{\bibinfo}[2]{#2}
\providecommand{\BIBentrySTDinterwordspacing}{\spaceskip=0pt\relax}
\providecommand{\BIBentryALTinterwordstretchfactor}{4}
\providecommand{\BIBentryALTinterwordspacing}{\spaceskip=\fontdimen2\font plus
\BIBentryALTinterwordstretchfactor\fontdimen3\font minus
  \fontdimen4\font\relax}
\providecommand{\BIBforeignlanguage}[2]{{%
\expandafter\ifx\csname l@#1\endcsname\relax
\typeout{** WARNING: IEEEtran.bst: No hyphenation pattern has been}%
\typeout{** loaded for the language `#1'. Using the pattern for}%
\typeout{** the default language instead.}%
\else
\language=\csname l@#1\endcsname
\fi
#2}}
\providecommand{\BIBdecl}{\relax}
\BIBdecl

\bibitem{reibman_tcomm93}
W.~Lam and A.~R. Reibman, ``Design of quantizers for decentralized estimation
  systems,'' \emph{IEEE Trans. Comm.}, vol.~41, no.~11, pp. 1602--1605, Nov.
  1993.

\bibitem{Marano&etal:07sp}
S.~Marano, V.~Matta, and P.~Willett, ``Asymptotic design of quantizers for
  decentralized \protect{MMSE} estimation,'' \emph{IEEE Trans. Signal
  Process.}, vol.~55, no.~55, pp. 5485--5496, Nov. 2007.

\bibitem{chen_tsp10}
H.~Chen and P.~K. Varshney, ``Performance limit for distributed estimation
  systems with identical one-bit quantizers,'' \emph{IEEE Trans. Signal
  Process.}, vol.~58, pp. 466--471, Jan. 2010.

\bibitem{chen_tsp10_2}
------, ``Nonparametric one-bit quantizers for distributed estimation,''
  \emph{IEEE Trans. Signal Process.}, vol.~58, no.~7, pp. 3777--3787, July
  2010.

\bibitem{wu:one-bit}
T.~Wu and Q.~Cheng, ``One-bit quantizer design for distributed estimation under
  the minimax criterion,'' in \emph{VTC Spring'10}, 2010, pp. 1--5.

\bibitem{kar_tsp12}
S.~Kar, H.~Chen, and P.~K. Varshney, ``Optimal identical binary quantizer
  design for distributed estimation,'' \emph{IEEE Trans. Signal Process.},
  vol.~60, no.~7, pp. 3896--3901, July 2012.

\bibitem{Tsitsiklis88}
J.~N. Tsitsiklis, ``Decentralized detection by a large number of sensors,''
  \emph{Math. Control Signals Systems}, vol.~1, no.~2, pp. 167--182, 1988.

\bibitem{swami_icisip05}
P.~Venkitasubramaniam, G.~Mergen, L.~Tong, and A.~Swami, ``Quantization for
  distributed estimation in large scale sensor networks,'' in \emph{Int. Conf.
  Intelligent Sensing and Information Process. (ICISIP 2005)}, Dec. 2005, pp.
  121--127.

\bibitem{xiao&etal:06spm}
J.~Xiao, A.~Ribeiro, Z.~Luo, and G.~Giannakis, ``{Distributed
  Compression-Estimation Using Wireless Sensor Networks},'' \emph{IEEE Signal
  Process. Mag., Special issue on Distributed Signal Processing for Sensor
  Networks}, vol.~23, no.~4, pp. 27--41, July 2006.

\bibitem{gubner_tit93}
J.~A. Gubner, ``Distributed estimation and quantization,'' \emph{IEEE Trans.
  Inf. Theory}, vol.~39, no.~4, pp. 1456--1459, Jul. 1993.

\bibitem{Chamberland&Veeravalli:SP03}
J.~Chamberland and V.~V. Veeravalli, ``{Decentralized Detection in Sensor
  Networks},'' \emph{IEEE Trans. Signal Process.}, vol.~51, pp. 407--416, Feb.
  2003.

\bibitem{gian_tsp06_1}
A.~Ribeiro and G.~B. Giannakis, ``Bandwidth-constrained distributed estimation
  for wireless sensor networks-part {I}: {G}aussian case,'' \emph{IEEE Trans.
  Signal Process.}, vol.~54, no.~3, pp. 1131--1143, Mar. 2006.

\bibitem{gian_tsp06_2}
------, ``Bandwidth-constrained distributed estimation for wireless sensor
  networks-part {II}: unknown probability density function,'' \emph{IEEE Trans.
  Signal Process.}, vol.~54, no.~7, pp. 2784--2796, Jul. 2006.

\bibitem{Vempaty_icassp13_opt}
A.~Vempaty, B.~Chen, and P.~K. Varshney, ``Optimal quantizers for distributed
  {B}ayesian estimation,'' in \emph{Proc. Int. Conf. Acoustics, Speech, and
  Signal Processing (ICASSP2013)}, Vancouver, Canada, May 2013, pp. 4893--4897.

\bibitem{Tsitsiklis:bookchapter}
J.~Tsitsiklis, ``Decentralized detection,'' in \emph{Advances in Statistical
  Signal Processing}, H.~Poor and J.~Thomas, Eds.\hskip 1em plus 0.5em minus
  0.4em\relax Greenwich, CT: JAI Press, 1993.

\bibitem{varshney_spmag06}
B.~Chen, L.~Tong, and P.~K. Varshney, ``Channel-aware distributed detection in
  wireless sensor networks,'' \emph{IEEE Signal Process. Mag. (Special Issue on
  Distributed Signal Processing for Sensor Networks)}, vol.~23, pp. 16--26,
  Jul. 2006.

\bibitem{Zamir98}
R.~Zamir, ``A proof of the {F}isher information inequality via a data processing
  argument,'' \emph{IEEE Trans. Inf. Theory}, vol.~44, no.~3, pp. 1246--1250,
  1998.

\bibitem{cover_inftheory}
T.~M. Cover and J.~A. Thomas, \emph{Elements of Information Theory}.\hskip 1em
  plus 0.5em minus 0.4em\relax Wiley, 1991.

\bibitem{vantrees_bounds}
H.~L.~V. Trees and K.~L. Bell, \emph{Bayesian Bounds for Parameter Estimation
  and Nonlinear Filtering/Tracking}.\hskip 1em plus 0.5em minus 0.4em\relax
  Piscataway, New Jersey: Wiley-IEEE Press, 2007.

\bibitem{brunt04}
B.~{V}an Brunt, \emph{The Calculus of Variations}.\hskip 1em plus 0.5em minus
  0.4em\relax New York: Springer, 2004.

\bibitem{osgood}
W.~F. Osgood, ``Sufficient conditions in the calculus of variations,''
  \emph{The Annals of Mathematics}, vol.~2, no. 1/4, 1900-1901.

\bibitem{kay_est}
S.~M. Kay, \emph{Fundamentals of Statistical Signal Processing Vol:I -
  Estimation Theory}.\hskip 1em plus 0.5em minus 0.4em\relax Prentice Hall,
  1993.

\bibitem{Chen_dep}
H.~Chen, B.~Chen, and P.~K. Varshney, ``{A New Framework for Distributed
  Detection With Conditionally Dependent Observations},'' \emph{IEEE Trans.
  Signal Process.}, vol.~60, no.~3, pp. 1409--1419, Mar. 2012.

\end{thebibliography}

\end{document}